\documentclass[12pt]{article}
\usepackage{dcolumn}
\usepackage{bm}
\usepackage{verbatim}       

\usepackage[margin=4.3cm]{geometry}

\usepackage[dvips]{graphicx}

\usepackage{amsthm}

\usepackage{bm}
\usepackage{amstext}
\usepackage{psfrag}
\usepackage{amsmath}
\usepackage{amsfonts}

\newcommand{\dd}{{\rm d}}

\newcommand{\bd}{\begin{definition}}                
\newcommand{\ed}{\end{definition}}                  
\newcommand{\bc}{\begin{corollary}}                 
\newcommand{\ec}{\end{corollary}}                   
\newcommand{\bl}{\begin{lemma}}                     
\newcommand{\el}{\end{lemma}}                       
\newcommand{\bp}{\begin{proposition}}            
\newcommand{\ep}{\end{proposition}}                
\newcommand{\bere}{\begin{remark}}                  
\newcommand{\ere}{\end{remark}}                     

\newcommand{\bt}{\begin{theorem}}
\newcommand{\et}{\end{theorem}}

\newcommand{\be}{\begin{equation}}
\newcommand{\ee}{\end{equation}}

\newcommand{\bit}{\begin{itemize}}
\newcommand{\eit}{\end{itemize}}
\newtheorem{theorem}{Theorem}[section]
\newtheorem{corollary}[theorem]{Corollary}
\newtheorem{lemma}[theorem]{Lemma}
\newtheorem{proposition}[theorem]{Proposition}
\theoremstyle{definition}
\newtheorem{definition}[theorem]{Definition}
\theoremstyle{remark}
\newtheorem{remark}[theorem]{Remark}

\begin{document}

\title{
Augustine of Hippo's philosophy of time  meets general relativity\thanks{This version includes the proof of Theorem 3.4 which is not included in the version  published in Kronoscope 14 (2014) 71-89. Previous title: Can God find a place in physics? St.~Augustine's philosophy
meets general relativity.}}

\author{E. Minguzzi\thanks{
Dipartimento di Matematica Applicata ``G. Sansone'', Universit\`a
degli Studi di Firenze, Via S. Marta 3,  I-50139 Firenze, Italy.
E-mail: ettore.minguzzi@unifi.it} }

\date{}

\maketitle

\begin{abstract}
\noindent A cosmological model is proposed which uses a causality
argument to solve the homogeneity and entropy problems of cosmology.
In this model a chronology violating region of spacetime causally
precedes the remainder of the Universe, and a theorem establishes
the existence of time functions precisely outside the chronology
violating region. This model is shown to nicely reproduce Augustine
of Hippo's thought on time and the beginning of the Universe. In the
model the spacelike boundary representing the Big Bang is replaced
by a null hypersurface  at which the gravitational degrees of
freedom are almost frozen while the matter and radiation content is
highly homogeneous and thermalized.
%
%

\end{abstract}

\section{Introduction}

In this work I shall present a cosmological model allowed by general
relativity which can be regarded as a mathematical representation of
previous ideas by 
Augustine of Hippo  (354 - 430)  on time and the creation of the
Universe. This model provides also some novel and natural solutions
to the homogeneity and entropy problems of cosmology.

Instead of introducing the model directly, in order to clarify the
correspondence with 
Augustine's thought,  we shall first investigate whether and how to
include a notion of God  in theoretical physics. Taking into account
that physics is expressed in the mathematical language we shall seek
a mathematical object that could represent at least some aspects of
what in common language we call God.

Of course, whatever definition of God is given it cannot be able to
cover all the  ideas that circulate concerning the nature of God.
For this essay I  shall start from a quite restrictive portion of 
Augustine philosophy in which  statements are made that may
allow us to identify a cosmological notion of God.

In order to avoid misunderstandings I stress that this paper is not
concerned about the problem of the existence of God, nor does it make
claims in this respect. Its main goal is to introduce some causality
arguments that may prove important in cosmology, and to stress the
amusing similarities with previous ideas by 
Augustine. Of course in order to make sense of those, one has to
assume, for the sake of the argument, the existence of God, for
otherwise it would be impossible to follow Augustine's philosophical
thought on the origin of time and the beginning of the Universe. In
particular, in this work we wish to clarify and emphasize
Augustine's amazing ability to distinguish between temporality and
causality, something that will find a mathematical expression only
with the advent of the general theory of relativity.

The reader interested on how modern cosmology has influenced
philosophers'
 arguments for God's existence may consult Craig and
Smith's book \cite{craig93} or the papers \cite{smith98,cahoone09}.
Craig's \cite{craig01} and Ganssle's \cite{ganssle01} monographs
give also a good account of time theories and how they relate to the
idea of God, in particular they treat the question of whether God
should consistently be thought as a temporal or timelessness entity,
a question that we shall also meet in what follows. Nevertheless,
the reader should keep in mind that those arguments may need to be
modified if the picture of the beginning of the Universe proposed in
these pages turns out to be correct, as the spacelike Big Bang
hypersurface would be replaced by a null hypersurface that continues
in a chronology violating region.



%
%
%
%

\section{Augustine's conclusions on the nature of God}


To begin we need to explore some ideas concerning the
nature of God. Here I shall consider some conclusions reached by the
philosopher 
Augustine that are largely independent of the
sacred texts 
and which are shared by different religions. They are:
\begin{enumerate}
\item[1.] There is an entity which we call { God} that satisfies the following points.
\item[2.] God has created the world.
\item[3.] God cannot be wrong.
\end{enumerate}
Our analysis will involve only these assumptions on the God side,
while for those on the scientific side we shall take from our
present knowledge of physics.

At the time of 
Augustine the Manich{\ae}ans asked the following
question "What was God doing before creating the world?".
Any answer to the question seems to involve a paradox. If God
created the world at one time and not at a previous time then God
changed his mind concerning the possibility of creating the world,
thus he was wrong in not creating it in the first place. The only
conclusion is that God cannot have created the world, because
whatever decision God takes he has already taken it.
This is a clear conflict with point 2 that states that God has
created the world, and thus that the world has not always existed.
The conclusion of the argument is in fact more general: the will of
God is eternal as there cannot be discontinuities in it, and so
should be all the creations that follow from that will.

Augustine's famous reply can be found in the XI book of the
{\em Confessions} \cite{augustine92,augustine07}. This book contains
one of the most fortunate studies of the concept of time especially
in the chapters starting from 14 where one can find the famous
sentence ``What, then, is time? If no one ask  me, I know; if I wish
to explain to him who asks, I know not.'' The reply that interests
us is contained in chapters 10-13 where he considers some issues
relating time, creation and God.

 First he states that he will not reply (Chap. 12, Par. 14)
\begin{quote}
``like the man that, they say, answered avoiding with a joke the
pressure of the question: `God was preparing the hell for those who
pry into such deep mysteries'. A thing is to understand, and another
thing is to jeer. I will not answer this way. I would more likely
answer: `I know not what I know not' [ ].''
\end{quote}


Augustine goes on clarifying that with {\em world} one must
understand all the creations of God. He accepts the conclusion that
the will of God is eternal, but denies that from that it follows the
eternity of the world. According to 
Augustine  all the times
are created by God itself so that God  comes ``before'' every time
although this ``before'' must  be understood in a causal but not in
a temporal way. In fact 
Augustine writes  (Chap. 13, Par. 16)

\begin{quote}
``It is not in time that you precedes the times. Otherwise you would
not precede them all. [ ]
You are always the same, your years never end. Your years neither come
nor go; ours instead come and go, for all of them will come. Yours
are all together because they are stable; they don't go because of
those coming, as they do not pass. Instead these, ours, will be when
all shall cease to be. Your years are one day, and your day is not
daily, but today; because your today yields no tomorrow, nor it
follows yesterday. Your day is the eternity [  ]. You created all
the times and before the times you were, and without a time there
wouldn't be any time''.
\end{quote}

Note that 
Augustine deduces, as the Manich{\ae}ans did, that
the will of God cannot change, but he does not find in that any
contradiction. For him, God does not perceive time as we do; not
only is God's will  in a kind of permanent state but it is its very
perception of time which shares this same permanence, this same
eternal state.

I regard 
Augustine reply to the Manich{\ae}an question as
logical given the premises. Of course although I claim that 
Augustine reply is logical I do not claim that  with these
considerations we are making science. Indeed, the main difficulty
relies in the quite unclear subjects and verbs entering points 1, 2
and 3. However, this problem cannot be avoided from the start. The
purpose of this work is to convert in a more rigorous language the
earlier sentences. For the moment let me
summarize what 
Augustine deduced from 1, 2 and 3 in the following additional
points.
\begin{itemize}
\item[4.] The will of God is eternal.
\item[5.] God created all the times, in particular God precedes all the times in
a causal way. Nevertheless, God does not precede the times in a
temporal way as the times did not exist before their creation.
\item[6.] Although God is not in our time, there is a kind of God's perception of time radically different from that of humans.
For God time is still, eternal, it is not perceived as a flow.
\end{itemize}
It is somewhat puzzling that 
Augustine used repeatedly the word
`times' in the plural form. Perhaps this is due to the fact that,
although we often regard the  Newtonian absolute time as the most
intuitive and widespread notion of time, it wasn't so for
Augustine. Another reason could be related with the concept
of psychological and hence subjective time that 
Augustine had
certainly elaborated (``Is in you, my mind, that I measure time''
Chap. 28, Par. 36). We shall return on the relevance of this maybe
accidental plurality later.

\section{The chronology protection conjecture}

Starting from assumptions 1, 2 and 3 we have been able to derive
further facts on God's nature given by points 4, 5 and 6. Despite
their somewhat vague formulations these conclusions will prove quite
stringent. Indeed, as we shall see, points 3, 4 and 6 will suggest
the mathematical object through which we could represent Augustine's
cosmological God, while points 1, 2 and 5 will allow us to put
further constraints on a Universe admitting a God. In particular
these constraints will offer new solutions to some old cosmological
problems.


We now need to assume some familiarity of the reader with general
relativity.  In short the spacetime $(M,g)$, is  a time oriented
4-dimensional manifold endowed with a Lorentzian metric of signature
$(-,+,+,+)$. The points of $M$ are called events. A non-vanishing
tangent vector $v\in TM_p$ can be spacelike, lightlike or timelike
depending on the value of $g(v,v)$ respectively, positive, zero or
negative. The lightlike directions give the direction of propagation
of light. The terminology extends to curves $\gamma: I\to M$,
provided the causal characterization of the tangent vector is
consistent throughout the curve. If there is a timelike curve
connecting two events $p$ and $q$ we write $p\ll q$ or $q\in I^+(p)$
or $(p,q)\in I^+$, where $I^+$ is the {\em chronology relation}. If
two events are connected by a causal curve or they are the same we
write $p\le q$ or $q\in J^+(p)$ or $(p,q)\in J^+$, where $J^+$ is
the {\em causal relation}.

%

It is  widely held that any reasonable spacetime should satisfy, along
with Einstein's equations, some additional causality requirement
\cite{hawking73}. One of the weakest requirements that can be
imposed on spacetime is that of chronology: there are no closed
timelike curves (sometimes called CTC).


The  fundamental problem of justifying chronology has received less
attention than deserved.
It is quite easy to construct solutions of the
Einstein equations that violate chronology, consider for instance
Minkowski spacetime with the slices $t=0$ and $t=1$ identified, or
think of G\"odel or Kerr's spacetimes. Thus the problem is not if
spacetime solutions of the Einstein equations can admit CTCs but
rather if reasonable spacetimes not presenting CTCs may develop
them.

S.\ Hawking argued that the laws of physics will always prevent a
spacetime to form closed timelike curves, in fact he raised this
expectation to the status of conjecture, now called {\em chronology
protection conjecture} \cite{thorne93}. According to it  the effects
preventing the formation of CTCs may also be quantistic in nature,
in fact Hawking claims that the divergence of the stress energy
tensor at the chronology horizon (i.e.\ the boundary of the {\em
chronology violating set}, the latter being  the region over which
CTCs pass) would be the principal candidate for a mechanism
preventing the formation of CTCs.

Despite some work aimed at proving the chronology protection
conjecture its present status remains quite unclear with some papers
supporting it and other papers suggesting its failure
\cite{tipler77,visser96,li96,li98,krasnikov02}. Some people think
that in order to solve the problem of the chronology protection
conjecture a full theory of quantum gravity would be required
\cite{hawking92,gott98}.

Apart from the technical motivations, the principal reason behind
the rejection of spacetimes presenting chronology violations remains
mostly a philosophical one. A closed timelike curve represents an
observer which is forced to live an infinite number of times the
same history (the grandfather paradox).


It is simply unacceptable that a human being, or any other entity
presenting some form of free will, be stuck into a cycle in which
always the same decisions are taken. Whatever a closed timelike
curve might represent there seems to be consensus that it cannot
represent the concept of ``observer'' to which we are used in
physics.

Nevertheless, whereas the usual notion of ``observer'' cannot be
represented by a CTC worldline, Augustine's cosmological God may
indeed be represented by such worldline. Indeed, we have seen that
according to middle-age philosophy  God has an eternal will (point
4) thus faced with the same conditions he would pass through the
same decisions. It cannot change direction because he  confirms the
correctness of the previous decision each time he is facing it.

It is curious that despite the fact that general relativity does not
model the concept of free will, the presence of CTC, by producing an
obstacle to this notion, provides a contact to philosophical
considerations that otherwise would be unrelated with this theory.

Now, we have to expand some more on the consideration that
Augustine's cosmological God may be modeled by a CTC. First recall
that the chronology violating set $\mathcal{C}$ is made of all the
points $p \in M$ such that $p\ll p$. This set splits into
equivalence classes $[p]$ by means of the equivalence relation
$p\sim q$ if $p\ll q\ll p$. In other words $p$ and $q$ belong to the
same equivalence class if there is a closed timelike curve passing
through both $p$ and $q$. Moreover, in this case the timelike curve
is entirely contained in $[p]$.
 It is possible to prove that the sets $[p]$
are all open in the manifold topology.

If $p$ and $q$ belong to the same chronology violating class then
they have  the same chronological role, in fact as $p \ll q$ and
$q\ll p$ it is not possible to say which one comes before. They are
in a sense `simultaneous'. Indeed, $p$ can be connected to $q$ also
by a lightlike causal curve and the same holds in the other
direction, thus it is indeed possible to move from $p$ to $q$ and
then from $q$ to $p$ in zero proper time. In particular any timelike
curve passing through $[p]$ would not cross events that follow `one
after the other' but rather almost equivalent events, actually
chronologically undistinguishable. This picture fits well with point
6, that is, with 
Augustine conclusion that ``Your years are one
day, and your day is not daily, but today; because your today yields
no tomorrow, nor it  follows yesterday. Your day is the eternity [
].'' All that suggests to regard God not as a single CTC, in fact
given one, one would get an infinite number of them in the same
chronology violating class, but rather as a chronology violating
class $[p]$ itself. This class $[p]$ has also to satisfy point 1,
which we convert into the mathematical statement $M=I^{+}([p])$,
namely any point of $M$ is chronologically preceded by a point of
God.

Thus we are led to the following definition

\begin{definition}
On a spacetime $(M,g)$,  we call {\em God} a chronology violating
class $[p]$ such that $M=I^{+}([p ])$.
\end{definition}
I will write this concept in italics in order to distinguish this
technical notion from Augustine's notion of cosmological God that we
met in the previous sections and that inspired it.

Note that given  a {\em God}, then any point of {\em God} generates
$M$ in the sense that $p \in \ God$ $\Rightarrow I^{+}(p)=M$, and
thus generates itself $p\ll p$. In suggestive terms, any portion of
{\em God} creates itself and the whole world.

Provided {\em God} exists it is unique, as the following theorem
proves
\begin{theorem}
There is at most one chronology violating class $[p]$ such that
$M=I^{+}([p])$.
\end{theorem}

\noindent {\em Proof}. Indeed,  $M= I^{+}([q])=I^{+}([p])$ implies
$q \in M= I^{+}([p])=I^{+}(p)$, and with the roles of $p$ and $q$
interchanged we get $p \in I^{+}(q)$, thus $p\sim q$ and hence
$[p]=[q]$. $\square$ \\

Since to any chronology violating class $[r]$ not satisfying
$M=I^{+}([r])$ we can still apply the arguments relating it to
points 4 and 6, we give the following definition

\begin{definition}
 We call {\em minor God} a chronology violating class which is not a
 {\em God}.
\end{definition}

Now, the chronology protection conjecture in its original
formulation may be rephrased as follows ``there are no {\em minor
Gods}'', in fact the chronology protection conjecture, roughly
speaking, states that chronology violating regions cannot form but
does not state that they cannot exist since the beginning of the
Universe. I must say, however, that  any mechanism accomplishing the
chronology protection would probably exclude, once applied to the
backward direction, also any chronology violating region. Probably
the issue as to whether there could be a mechanism that removes {\em
minor Gods} while keeping a {\em God} could be answered only by
 showing the details of the chronology protection mechanism.

Let us assume for simplicity that there are no {\em minor Gods} and
let us show in which way the definition of {\em God} satisfies point
5. Recall that a {\em time function} is a continuous function $t:
M\to \mathbb{R}$ such that $x<y \Rightarrow t(x) < t(y)$, namely a
function that increases over every causal curve. For instance any
observer in Minkowski spacetime has its own time function.

Clearly, no time function can exist in the presence of a CTC,
because if $p \ll q\ll p$ then $t(p) <t(q) < t(p)$, which is
impossible. Indeed, the presence of a time function is equivalent to
{\em stable causality} (i.e.\ causality is stable under sufficiently
small perturbations of the metric) which is a much stronger
causality property than chronology. Given one time function one has
that a multitude of time functions exist.

Nevertheless, although $M$ does not admit a time function, the
spacetime $M\backslash \bar{\mathcal{C}}$ with the induced metric
may indeed admit a time function and hence many of them. In other
words, the part of spacetime not containing {\em God} (or better its
closure) may admit time functions. In this sense {\em God} precedes
the region of the Universe were time makes sense, but in a causal
rather that a temporal way as those time functions are not defined
in the region of {\em God}. This is exactly 
Augustine's
conclusion summarized by point 5.

The nice fact is that not only $M\backslash \bar{\mathcal{C}}$ {\em
may} admit a time function, but that it {\em must} admit a time
function, provided null geodesic completeness and other reasonable
physical conditions are satisfied. For more details on these
conditions see \cite{minguzzi08d,hawking73}.


\begin{theorem} \label{vgs}
Let $(M,g)$ be a spacetime which admits  no chronology violating
class but possibly for the one, denoted $[r]$, which generates the
whole universe, i.e. $I^{+}([r])=M$. Assume that the spacetime
satisfies the null convergence condition and the null genericity
condition on the lightlike inextendible geodesics which are entirely
contained in $M\backslash\overline{[r]}$, and suppose that these
lightlike geodesics are complete. Then the spacetime $M\backslash
\overline{[r]}$ is  stably causal and hence admits a time function.
\end{theorem}

(For the proof see the appendix.)

The fact that the assumption of null geodesic completeness may be
actually compatible with the singularity theorems is discussed
in \cite{minguzzi08d}.

In conclusions we have given a definition of God that satisfies some
technical properties which represent pretty well points 2-6.

The figure \ref{god} summarizes the picture of a spacetime admitting
a {\em God}. There are in fact solutions of the Einstein equations
admitting a similar causal structure. The most important is the
Taub-NUT metric, which so far has not been considered as a serious
candidate for a cosmological solution. Here I would like to suggest
that if not the metric structure,  at least the causal structure of
the Taub-NUT solution  could indeed be similar to that of our
Universe. In fact sometimes causal structures like Taub-NUT are
dismissed on the ground that they have no `Big Bag', no initial
singularity, a fact which would contradict Hawking's singularity
theorem and observations.

This conclusion is incorrect: Hawking's (1967) singularity theorem
states that, given an expanding cosmological flow and some other
conditions, there should be some past incomplete timelike geodesic.
However, this timelike geodesic may well be totally imprisoned in a
compact set. In this case it may spiral towards the boundary of the
chronology violating set without reaching it. In this picture the
`Big Bang' is replaced by the boundary of the chronology violating
region, exactly that slice that separates {\em God} from the rest of
the Universe. Finally, its hot nature seems to fit well with the
said divergence of the stress energy tensor that is expected
according to the chronology protection conjecture. In fact there is
also the possibility that the matching between the chronology
violating set and the rest of the universe be accomplished up to a
singular scale transformation. In this case the causal structure
would be perfectly meaningful as a whole but the metric would not as
it could not be continued through the boundary. For more details on
these extension techniques see \cite{lubbe07}.

I conclude that it is possible to conceive a reasonable Universe
whose causal structure has features analogous to Taub-NUT (Misner)
and that then, after an initial phase, has the light cones tilted to
match an expanding FLRW Universe. Spacetimes presenting some of
these elements are for instance the $\lambda$-Taub-NUT spacetimes.

Similar models have already appeared in the literature. An important
article that anticipated some ideas considered in this work is
\cite{gott98}. However, while in that article the authors focused on
the problem of quantum field theory in spacetime with CTCs, here I
shall consider mainly the problems of homogeneity and entropy and
their relation with causality. In particular, I will introduce the
idea of the rigidity at the boundary of the chronology violating
region (see next section). R.\ Penrose \cite{penrose08} has also
advocated the possibility that the Big Bang could be only a layer
separating our observed Universe from a previous stage. An essential
difference with this proposal is the fact that he keeps a spacelike
Big Bang boundary, while as I shall explain, the null boundary
allows us to solve the homogeneity and entropy problems in a much
more natural way. Moreover, while he has to work with a cyclic
cosmology in order to satisfy the Weyl tensor hypothesis, we do not
need such constraint since we can justify Weyl tensor hypothesis
using the null character of the Big Bang.

\begin{figure}[ht]
\begin{center}
 \includegraphics[width=10cm]{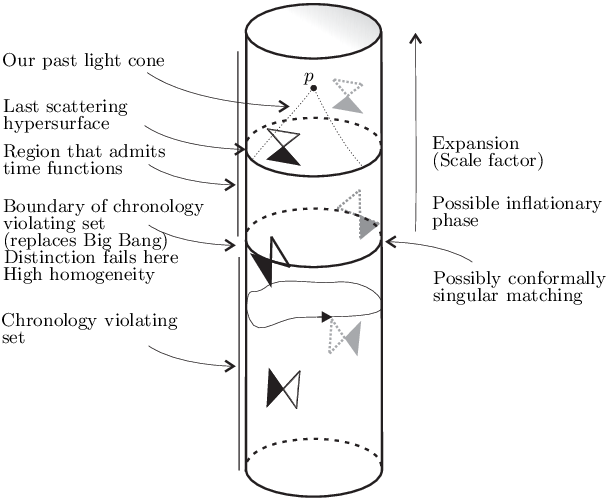}
\end{center}
\caption{A Universe with $S^1$ section which gives an idea of the
cosmological picture presented in this work. The region that admits
time functions is causally preceded by the chronology violating set
({\em God}) as in 
Augustine's conclusions.} \label{god}
\end{figure}

\section{The homogeneity
and entropy problems}

The cosmic microwave background (CMB) formed when, after a sufficient expansion of the Universe, the density of matter decreased to a level that light decoupled from it (the mean free path of light became infinite). The set of events of departure of those photons form an ideal last scattering hypersurface. Today we observe just a portion of that hypersurface, namely the intersection of it with our past light cone. Since we observe that the CMB radiation has the same spectra (temperature) independently of the direction of observation, there is the problem of justifying such isotropy on the night sky. In fact the regions that emitted that radiation were so far apart that, according to the FLRW scenario, they didn't have any past point in common. This is the isotropy or homogeneity problem depending on whether one refers to the isotropy of temperature on the night sky or on the equivalent homogeneity of temperature on the last scattering hypersurface.

It is often claimed that inflation solves this problem. The idea is
that if a patch of space expands so much, in the initial phase of
the Universe, to include the whole surface we see today, then it
should be natural to observe homogeneity. This argument works only
if homogeneity is assumed at a different scale, actually at a much
smaller scale, prior to inflation, namely if the initial patch is
considered homogeneous.

Indeed, if we look closer and closer at, say, a crystal of salt,
although apparently homogeneous it will show atomic inhomogeneities
when zoomed sufficiently, that is, expansion produces homogeneity
only if homogeneity is already present at a much smaller scale.

This criticism has been moved to inflation by several authors, as
rather than solving the problem of homogeneity, inflation seems to
replace a type homogeneity assumption with another
\cite{calzetta92,goldwirth92,cornish96,penrose05}. R.\ Penrose
argues that inflation may well prove to be correct but not for the
initial arguments moved in its favor \cite{penrose05}.

Instead, the assumption that there is a chronology violating region
generating the whole Universe explains rather easily the homogeneity
of the CMB radiation. Indeed, the explanation has nothing to do with
the expansion of the Universe (namely to  the conformal scale
factor) but rather to its causal structure. In our model any point
$p$ in the last scattering hypersurface contains, in its own past,
the chronology violating region $[r]$, namely $[r] \subset I^{-}(p)$
and in fact its boundary. Thus the chronological pasts of the points
in the last scattering hypersurface share many points on spacetime,
and thus it is reasonable that they have similar temperatures.

Let us now make a few comments that will be useful in the discussion of the entropy problem. We have justified the isotropy of the CMB radiation showing  that the points in the last scattering hypersurface have chronological pasts which share the boundary of the chronology violating region (the new proposal for 'Big Bang').  Instead, in the inflationary picture it was assumed that the homogeneous temperature  at the last scattering hypersurface was attained through a process called thermalization according to which causally disconnected regions at the Big Bang came into contact reaching a common temperature before decoupling between matter and radiation (last scattering hypersurface). R. Penrose \cite{penrose79,penrose05}  has pointed out that the thermalization mechanism cannot be considered a satisfactory explanation for homogeneity as it conflicts with the so called entropy problem to which we shall return in a moment.

Our solution to the isotropy problem of the CMB radiation does not require thermalization because the almost constant temperature on the last scattering hypersurface comes from the fact that these points share most of their chronological past independent of whether the events on their past  have the same temperature. Moreover, it seems likely that, according to our proposal of causal structure, the universe would be already at a very homogeneous state at the boundary of the chronology violating region ('Big Bang'). Indeed, some mathematical results, connected with the concept of compactly generated Cauchy horizon and imprisoned curves \cite{krolak04,minguzzi07f}, suggest that this boundary must be generated by lightlike geodesics whose closure is exactly the boundary (as it happens in figure \ref{god}). A well known open conjecture states that the compact Cauchy horizons under positivity of energy (e.g., null energy condition) are differentiable, where the generators do not escape the boundary neither in the future nor in the past direction. The existence of the generators implies that these horizons admit global vector fields, a fact that constrains their topology (for instance, in 2+1 dimension it must be a torus while in 3+1 dimensions there are more possibilities). These results  are also expected to hold for the compact boundaries of chronology violating sets since their mathematical properties are similar to those of Cauchy horizons.

Given any two points
on the boundary $p,q$ one would have $q \in \overline{I^{+}(p)}$ and
$p \in \overline{I^{+}(q)}$, thus in practice they could be
considered as causally related. As they can communicate through the
boundary, this boundary is expected to attain an homogeneous
temperature prior to any subsequent expansion. Since the development of Quantum Field Theory under CTC is at a early state of development, this claim is only a speculation. The analogy is that of a quantum field over a torus  where the closed geodesic generators of the torus are replaced by the generators of the Cauchy horizon. The boundary conditions allow us to expand it in the Fourier modes, and the most relevant one is that of lower excitation for which the field is constant, namely homogeneous.

Of course this mechanism may be followed by that of inflation, but
we point out that it does not seem to be necessary. Indeed, the main
accomplishment of inflation seems to be its ability to predict the
correct density inhomogeneities over the homogeneous background.
Hollands and Wald \cite{hollands02} have recently argued  not
only that inflation does not satisfactorily solve the homogeneity problem
but also that the desired scale free spectra of the perturbations
can be obtained even in the absence of inflation. They therefore claim
that the main problem is that of homogeneity/isotropy as they could
not find any dynamical mechanism for it. We argued that such a
mechanism exists, the solution lies in assuming the existence of a
chronology violating region from which the Universe develops: a {\em
God} in our terminology.

\subsection{The entropy problem and the rigidity of achronal
hypersurfaces}

Let us come to the entropy problem. This difficulty of standard
cosmology arises when considering the huge difference between the
entropy of the Universe today with that at the time of the Big Bang.
R.\ Penrose by taking into account also the gravitational entropy,
has argued that the Universe at its beginning had probably to be
thermalized, to account for the homogeneity problem, but
nevertheless it had to be special as the calculation of the entropy
shows that it was much smaller than today.

 Penrose concludes that the gravitational degrees of
freedom had to be in a very special state. In his view the Universe
could increase in entropy despite its initial thermalization because
in the beginning the gravitational degrees of freedom were almost
frozen.


By the way, Penrose reaches this conclusion clarifying a common
misconception that attributes the initial small entropy of the
Universe to its size. Penrose shows that this position is untenable
by considering potentially recontracting universes.

He also notes that when  matter is left to the action of gravity it
tends to clump, passing from an homogeneous state to an
inhomogeneous one. The Weyl tensor increases because of this
clumping, and therefore this tensor may quantify in some sense the
amount of entropy contained in the gravitational degrees of
freedoms. Thus Penrose ends suggesting that in the beginning of the
universe the Weyl tensor had to be very small, and possibly zero.
This is Penrose's Weyl tensor hypothesis \cite{penrose79,goode91}.
We note that the important point is not that the Weyl tensor be zero
but rather that its components be  fixed as this seems to be enough
to guarantee that the gravitational degrees of freedom were
initially frozen.

In order to avoid misunderstandings we stress that the isotropy of the CMB radiation does not provide evidence for a vanishing Weyl tensor since decoupling, since this homogeneity is observed at a large cosmological scale and hence holds only for the averaged metric. The real Weyl tensor is non-zero at a local scale due to clumping, and its cumulated effects over a space section is non-vanishing.  The Penrose's Weyl tensor hypothesis states that this local cumulated contribution is zero at the Big Bang while it is far from zero at the present epoch and at decoupling.

The picture of the beginning of the Universe presented in this work
is likely to satisfy the Penrose's Weyl tensor hypothesis. Indeed,
as I mentioned, the boundary of the chronology violating region
would be generated by lightlike geodesics (which are moreover
achronal). Now, there is a {\em rigidity result} \cite{beem96} which
states that an  asymptotically simple vacuum spacetimes is isometric
to Minkwoski spacetime in a neighborhood of every achronal lightlike
geodesic (Galloway's null splitting theorem \cite{galloway00}). I
expect that analogous results should hold for the case considered in
this work, that is, I expect the spacetime near the boundary of the
chronology violating region to be isometric to some highly symmetric
spacetime. This rigidity would clearly  fix the Weyl tensor and thus
send to zero the degrees of freedom contained in it.

In order to grasp  why an achronal boundary generated by lightlike
geodesics is able to fix the geometry  by constraining the Weyl
tensor, consider the equations for the expansion $\theta$ (measuring
the divergence of the transverse section to the flow) and for the
shear $\sigma_{m n}$ (measuring the deformation) of the geodesics
running on such hypersurface \cite[Sect. 4.2]{hawking73}
\begin{align*}
\frac{\dd \theta}{\dd v}  &=- R_{a b} K^a K^b -2 \sigma^2
-\frac{1}{2} \theta^2,  \\
\frac{\dd \sigma_{m n}}{\dd v}&=-C_{m4n4} - \theta \sigma_{m n}-
\sigma_{m p} \sigma_{p n}+\delta_{m n} \sigma^2. \qquad (\textrm{sum
over p})
\end{align*}
where $m,n=1,2$, $K^a$ is the tangent vector to the null congruence,
$R_{a b}$ is the Ricci tensor, $C_{a b c d}$ is the Weyl tensor, $v$
is the affine parameter and a base adapted to the congruence has
been chosen.  Under future null completeness the first equation
implies the presence of a focusing point if $R_{a b} K^a K^b + 2
\sigma^2>0$ at some point of the null hypersurface. However, the
presence of a focusing point would contradict achronality thus
$\sigma_{m n}=0$ everywhere  and from the second equation one gets
that $C_{m4n4}=0$ all over the null hypersurface, a fact which may
be regarded as a partial confirmation of Penrose's Weyl tensor
hypothesis. Actually one can dispense with the assumption of future
null geodesic completeness on the null hypersurface provided this
hypersurface is compact. The idea is that if it would not hold then
by adapting Prop.\ 6.4.4 of \cite{hawking73} to the almost closed
case one could infer the presence of a closed  timelike curve in the
future of the null hypersurface contradicting the assumption that
this hypersurface bounds the only chronology violating class.

In short we have given an argument that supports the Weyl tensor conjecture and a solution to the entropy problem compatible with the arguments originally proposed by Penrose.

Clearly, a full proof will require further study since it necessary to show that the degrees of freedom of the Weyl tensor transverse to the Big Bang null hypersurface vanish. Still the mechanism proposed in this work is particularly effective in sending some components to zero, a fact that might indeed point to the validity of our physical assumption: that the Big Bang hypersurface is lightlike rather than spacelike.

We end by observing that the null genericity condition would not
hold for geodesics lying on the boundary of the chronology violating
region. Fortunately we do not need it in theorem \ref{vgs}, hence
its consequences are consistent with the rigidity of the boundary.

\section{Conclusions}

In this work I presented a picture for the beginning of the Universe which seems to be able to solve the isotropy and entropy problems. In essence the Big Bang has to be replaced with a null hypersurface such that all the points on it have the same chronological future (i.e. future distinction is violated). As a consequence, the points in the last scattering (spacelike) hypersurface have chronological pasts that contain one and hence all points of this null hypersurface, a fact that clarifies the observed temperature homogeneity.

While the solution of the isotropy problem is rather straightforward under this assumption, the solution to the entropy problem relies on some conjectures and speculations. In part, this is unavoidable since we follow Penrose's idea, which relates the entropy of the universe with the Weyl tensor. The existence of such relation is itself a strong speculation though supported by physical arguments. Using this idea, and assuming that  compact Cauchy horizons are differentiable, we were  able to infer the result that  the Big Bang, interpreted as a past Cauchy horizon, is generated by inextendible lightlike geodesics and hence that some components of the Weyl tensor vanish there, thus supporting Penrose's strategy of solution to the entropy problem.

If proved necessary, the just mentioned beginning of the Universe
may be followed by a period of inflation, so that it is indeed
possible to join the good accomplishments of inflation with the
solution of the homogeneity and entropy problems given by the above
idea.

By a stability argument, the spacetime once continued through the
null hypersurface must develop closed timelike curves. Indeed, a
spacetime in which the cones tilt in the opposite sense
 would have a null hypersurface (and hence a
failure of distinction) that disappears under a small perturbation
of the metric (see figure \ref{comp} or figure 37 of
\cite{hawking73}).

\begin{figure}[ht]
\begin{center}
 \includegraphics[width=10cm]{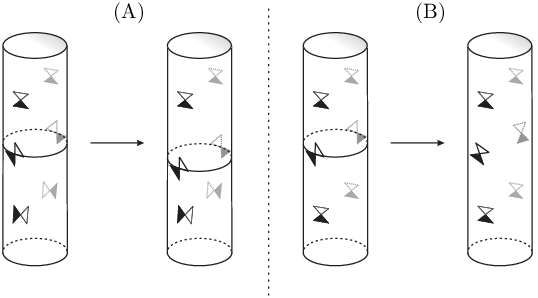}
\end{center}
\caption{The presence of the closed null hypersurface in place of
the spacelike Big Bang hypersurface allows us to give a causality
solution to the entropy and homogeneity problems of cosmology. A
stability argument shows that below the null hypersurface there must
be a chronology violating region. Indeed, in case (A) by a small
perturbation of the metric near its boundary the null hypersurface
moves up or down but does not disappear, hence the argument holds
true even after a small perturbation. Instead, in case (B) in which
there is no chronology violation, a small tilting of the light cones
in the forward direction near the hypersurface destroys the null
hypersurface. It must be noted that it is meaningless to make
(quantum) perturbation theory near the boundary, because any
perturbation pass to a perturbation of the metric and thus moves the
boundary itself.} \label{comp}
\end{figure}

With the aim of solving the homogeneity and entropy problem one is
therefore naturally led to the idea of a chronology violating region
from which the whole Universe has developed \cite{gott98}.

I showed that this picture for a Universe fits  well with some
conclusions reached by 
Augustine while he was answering some
questions raised by the Manich{\ae}ans. To appreciate the
correspondence it is necessary to identify God with the chronology
violating set that precedes the whole Universe.

I must say that I was developing the physical content of this work
before discovering 
Augustine thought in the  {\em Confessions}.
Nevertheless, I was so puzzled  by the correspondence that decided
to present them in conjunction so as to stress the similarities.
While doing so I discovered some unexpected results like theorem
\ref{vgs} which I missed in previous analysis of similar problems, a
fact which to my mind made the correspondence even more interesting.

One may ask how it happened that 
Augustine went so close to the
model of Universe presented in these pages, given that he certainly
ignored general relativity. My own opinion is that while one is thinking
about a subject there are many ways of coming to trivial or
incorrect conclusions, whereas only a few paths can lead to correct
or at least interestingly structured thoughts. It is therefore not
an accident that 
Augustine deep reflections on time, creation
and God can find today a correspondence in general relativity. It
should suffice to consider that the latter is the most advanced
theory we have ever had on the dynamics of time.

\section*{Acknowledgments}
Work presented at the 5th Iberian Cosmology Meeting held in Porto,
March 29-31, 2010. I would like to thank the colleagues who have
encouraged me in pursuing the investigation of this alternative
model for the beginning of the Universe. This work has been
partially supported by GNFM of INDAM and by FQXi.

\subsection*{Appendix: Proof of Theorem \ref{vgs}.}

Recall that a future lightlike ray is a future inextendible achronal
causal curve, in particular it is a lightlike geodesic. Past
lightlike rays are defined analogously. A lightlike line is an
achronal inextendible causal curve. In particular a lightlike line
is a lightlike geodesic without conjugate points. The boundary of a
set is denoted with a dot.

For the proof of the next lemma see \cite[Prop.\ 2]{kriele89}, or
the proof of \cite[Theorem 12]{minguzzi07d}.

\begin{lemma} \label{vhq}
Let $[r]$ be a chronology violating class. If $p \in \dot{[r]}$ then
through $p$ passes a future lightlike ray contained in $\dot{[r]}$
or a past lightlike ray contained in $\dot{[r]}$ (and possibly
both).
\end{lemma}

\begin{lemma} \label{kjh}
Let $[r]$ be a chronology violating class such that $I^{+}([r])=M$
then $\dot{[r]}$, is generated by future lightlike rays contained in
$\dot{[r]}$ and $J^{-}(\overline{[r]}\,)=\overline{[r]}$.
\end{lemma}

\noindent {\em Proof}. Let $p \in \dot{[r]}$ then since $p \in
M=I^{+}([r])$ it cannot be $p\in I^{-}([r])$ otherwise $p\ll r\ll
p$, i.e.\ $p\in [r]$, a contradiction. As $p\in \dot{[r]} \backslash
I^{-}([r])$, there is a sequence $p_n \in [r]$, $p_n \to p$. Since
$p_n \in [r]$ there are timelike curves $\sigma_n$ entirely
contained in $[r]$ which connect $p_n$ to $r$. By the limit curve
theorem \cite{minguzzi07c} there is either (a) a limit continuous
causal curve connecting $p$ to $r$, in which case as $[r]$ is open,
$p \in I^{-}([r])$, a contradiction, or (b) a limit future
inextendible continuous causal curve $\sigma$ starting from $p$ and
contained in $\overline{[r]}$. Actually $\sigma$ is contained in
$\dot{[r]}$ otherwise $p \in I^{-}([r])$, a contradiction. Moreover,
$\sigma$ is a future lightlike ray, otherwise there would be $q \in
\dot{[r]}\cap \sigma$, $p\ll q$ and as $I^{+}$ is open $p \in
I^{-}([r])$, a contradiction.

For the last equality, assume by contradiction, $q \in
J^{-}(\overline{[r]})\backslash\overline{[r]}$. Since $q \in
M=I^{+}(r)$ there is a timelike curve joining $r$ to $q$ and a
causal curve joining $q$ to $\overline{[r]}$. By making a small
variation starting near $q$ we get a timelike curve from $r$ to
$\overline{[r]}$, and hence equivalently, from $r$ to $r$ passing
arbitrarily close to $q$, thus $q \in \overline{[r]}$, a
contradiction. $\square$ \\

Here I give the proof of theorem \ref{vgs}. It is a non-trivial
generalization over the main theorem contained  in
\cite{minguzzi07d}. \\


 \noindent {\em Proof}. Consider the spacetime $N=M\backslash
\overline{[r]}$ with the induced metric, and denote by $J^{+}_N$ its
causal relation. This spacetime is clearly chronological and in fact
strongly causal. Indeed, if strong causality would fail at $p \in N$
then there would be sequences $p_n, q_n \to p$, and causal curves
$\sigma_n$ of endpoints $p_n, q_n$, entirely contained in $N$, but
all escaping and reentering some neighborhood of $p$. By an
application of the limit curve theorem \cite{minguzzi07c,beem96} on
the spacetime $M$ there would be an inextendible continuous causal
curve   $\sigma$ passing through $p$ and contained in $\bar{N}$ to
which a reparametrized subsequence $\sigma_n$ converges uniformly on
compact subsets ($\sigma$ can possibly be closed). The curve
$\sigma$  must be achronal otherwise one would easily construct a
closed timelike curve intersecting $N$ (a piece of this curve would
be a segment of some $\sigma_n$ thus intersecting $N$). Thus
$\sigma$ is a lightlike line. If the line is entirely contained in
$N$ then it is inextendible in the spacetime $(N,g_N)$ and being
complete by assumption, since null genericity and null convergence
hold, there would be two conjugate points, a fact which contradicts
the achronality of $\sigma$.


The possibility that $\sigma$ intersects $\dot{N}=\dot{[r]}$ leads
also to a contradiction because  $\sigma$ cannot intersect
$\dot{[r]}$ in the causal future of $p$ as
$J^{-}(\overline{[r]}\,)=\overline{[r]}$. But if it intersects
$\dot{[r]}$ in the causal past of $p$, $\sigma$ cannot be tangent to
the generators of $\dot{[r]}$ otherwise by Lemma \ref{kjh}  $p \in
\dot{[r]}$, a contradiction, thus after the intersection with
$\dot{[r]}$ (again by Lemma \ref{kjh}) $\sigma$ remains in
$\dot{[r]}$. This fact implies that $\sigma$ is not achronal, as it
has a corner, a contradiction. We conclude that strong causality
holds on $N$.
%
%

 The next step is to prove  that $\overline{J^{+}_N}$
is transitive. In this case $N$ would be causally easy
\cite{minguzzi08b} and hence stably causal (thus admitting time
functions). The transitivity of $\overline{J^{+}_N}$ is proved as
done in \cite[Theorem 5]{minguzzi07d}, the only difference is that
the argument allows us only to prove that if, $x,y, z \in N$, $(x,y)
\in \overline{J^{+}_N}$ and $(y,z) \in \overline{J^{+}_N}$ then
$(y,z) \in \overline{J^{+}}(=\overline{I^{+}})$ as the limit causal
curve passing through $y$ may intersect $\dot{N}$. However, there
are neighborhoods $U$ and $V$ such that any timelike curve
connecting $U\ni x$, $U \subset N$ to $V \ni z$, $V \subset N$ must
stay in $N$, because otherwise there would be some $w \in
\overline{[r]}$ such that $x' \le w$, with $x' \in U$. This is
impossible because by proposition \ref{kjh},
$J^{-}(\overline{[r]}\,) \subset \overline{[r]}$.  $\square$ \\



\begin{thebibliography}{10}

\bibitem{beem96}
Beem, J.~K., Ehrlich, P.~E., and Easley, K.~L.: \emph{Global
Lorentzian
  Geometry}.
\newblock New York: Marcel {D}ekker {I}nc. (1996)

\bibitem{cahoone09}
Cahoone, L.: Arguments from nothing: {G}od and quantum cosmology.
\newblock Zygon \textbf{44}, 777--796 (2009)

\bibitem{calzetta92}
Calzetta, E. and Sakellariadou, M.: Inflation in inhomogeneous
cosmology.
\newblock Phys. Rev. D \textbf{45}, 2802--2805 (1992)

\bibitem{cornish96}
Cornish, N.~J., Spergel, D.~N., and Starkman, G.~D.: Does chaotic
mixing
  facilitate {$\Omega<1$} inflation?
\newblock Phys. Rev. Lett. \textbf{77}, 215--218 (1996)

\bibitem{craig01}
Craig, W.~L.: \emph{Time and eternity: Exploring God's relationship
to time}.
\newblock Wheaton, Illinois: Crossway {B}ooks (2001)

\bibitem{craig93}
Craig, W.~L. and Smith, Q.: \emph{Theism, Atheism and {B}ig {B}ang
Cosmology}.
\newblock Oxford: Oxford {U}niversity {P}ress (1993)

\bibitem{galloway00}
Galloway, G.~J.: Maximum principles for null hypersurfaces and null
splitting
  theorems.
\newblock Ann. {H}enri {P}oincar\'e \textbf{1}, 543--567 (2000)

\bibitem{ganssle01}
Ganssle, G.~E.: \emph{God \& {T}ime: 4 Views}.
\newblock Downers Grove: Inter{V}arsity {P}ress (2001)

\bibitem{goldwirth92}
Goldwirth, D.~S. and Piran, T.: Initial conditions for inflation.
\newblock Phys. {R}ep. \textbf{214}, 223--292 (1992)

\bibitem{goode91}
Goode, S.~W.: Isotropic singularities and the {P}enrose-{W}eyl
tensor
  hypothesis.
\newblock Class. Quantum Grav. \textbf{8}, L1--L6 (1991)

\bibitem{gott98}
{Gott III}, J.~R. and Li, L.-X.: Can the universe create itself?
\newblock Phys. Rev. D \textbf{58}, 023501 (1998)

\bibitem{hawking92}
Hawking, S.~W.: Chronology protection conjecture.
\newblock Phys. Rev. D \textbf{46}, 603--611 (1992)

\bibitem{hawking73}
Hawking, S.~W. and Ellis, G. F.~R.: \emph{The Large Scale Structure
of
  Space-Time}.
\newblock Cambridge: Cambridge {U}niversity {P}ress (1973)

\bibitem{hollands02}
Hollands, S. and Wald, R.~M.: An alternative to inflation.
\newblock Gen. Relativ. Gravit. \textbf{34}, 2043--2055 (2002)

\bibitem{krasnikov02}
Krasnikov, S.: No time machines in classical general relativity.
\newblock Class. Quantum Grav. \textbf{19}, 4109–4129 (2002)

\bibitem{kriele89}
Kriele, M.: The structure of chronology violating sets with compact
closure.
\newblock Class. Quantum Grav. \textbf{6}, 1607--1611 (1989)

\bibitem{krolak04}
Kr{\'o}lak, A.: Cosmic censorship hypothesis.
\newblock Contemporary Mathematics \textbf{359}, 51--64 (2004)

\bibitem{li96}
Li, L.-X.: Must time machine be unstable against vacuum
fluctuations?
\newblock Class. Quantum Grav. \textbf{13}, 2563--2568 (1996)

\bibitem{li98}
Li, L.-X. and {Gott III}, J.~R.: Self-consistent vacuum for {M}isner
space and
  the chronology protection conjecture.
\newblock Phys. Rev. Lett. \textbf{80}, 2980--2983 (1998)

\bibitem{lubbe07}
L{\"u}bbe, C. and Tod, P.: An extension theorem of conformal gauge
  singularities (2007).
\newblock ArXiv:0710.5552v2

\bibitem{minguzzi07c}
Minguzzi, E.: Limit curve theorems in {L}orentzian geometry.
\newblock J. Math. Phys. \textbf{49}, 092501 (2008)

\bibitem{minguzzi07f}
Minguzzi, E.: Non-imprisonment conditions on spacetime.
\newblock J. Math. Phys. \textbf{49}, 062503 (2008)

\bibitem{minguzzi07d}
Minguzzi, E.: Chronological spacetimes without lightlike lines are
stably
  causal.
\newblock Commun. Math. Phys. \textbf{288}, 801--819 (2009)

\bibitem{minguzzi08b}
Minguzzi, E.: {$K$}-causality coincides with stable causality.
\newblock {C}ommun. {M}ath. {P}hys. \textbf{290}, 239--248 (2009)

\bibitem{minguzzi08d}
Minguzzi, E.: On the global existence of time.
\newblock Int. {J}. {M}od. {P}hys. {D} \textbf{18}, 2135--2144 (2009).
\newblock Third juried prize at the FQXi contest on the `Nature of Time'.

\bibitem{penrose79}
Penrose, R.: \emph{Singularities and time-asymmetry}, Cambridge:
Cambridge
  {U}niversity {P}ress, vol. General relativity: {A}n {E}instein centenary
  survey, pages 581--638 (1979)

\bibitem{penrose05}
Penrose, R.: \emph{The road to reality: A complete guide to the laws
of the
  Universe}.
\newblock New York: A. A. Knopf (2005)

\bibitem{penrose08}
Penrose, R.: \emph{Causality, quantum theory and cosmology},
Cambridge:
  Cambridge {U}niversity {P}ress, vol. On space and time, pages 141--195 (2008)

\bibitem{augustine07}
{Sant'Agostino}: \emph{Il tempo}.
\newblock Piccola {B}iblioteca {A}gostiniana. Roma: Citt{\`a} {N}uova (2007).
\newblock A cura di {G}. {C}atapano

\bibitem{smith98}
Smith, Q.: Why {S}tephen {H}awking's cosmology precludes a creator.
\newblock Philo: {A} {J}ournal of {P}hilosophy \textbf{1}, 75--93 (1998)

\bibitem{augustine92}
{St. Augustine}: \emph{Confessions}.
\newblock World's {C}lassics. Oxford: Oxford {U}niversity {P}ress (1992).
\newblock Translated by {H}. {C}hadwick

\bibitem{thorne93}
Thorne, K.: \emph{Closed Timelike Curves}, Bristol, England:
Institute of
  Physics Publishing, vol. General Relativity and Gravitation, pages 295--315
  (1993)

\bibitem{tipler77}
Tipler, F.~J.: Singularities and causality violation.
\newblock Ann. {P}hys. \textbf{108}, 1--36 (1977)

\bibitem{visser96}
Visser, M.: \emph{Lorentzian Wormholes}.
\newblock New York: {Springer-Verlag} (1996)

\end{thebibliography}


\end{document}